\documentclass{article}
\usepackage{amsmath, amsthm}
\usepackage{xcolor}
\usepackage[top=15mm, bottom=20mm, left=30mm, right=30mm]{geometry}
\usepackage{appendix}
\usepackage{graphicx}
\usepackage{enumerate}
\usepackage{longtable}

\newtheorem{defn}{Definition}[section]

\newtheorem{prop}[defn]{Proposition}

\newtheorem{eg}[defn]{Example}

\newcommand{\agents}{\mathcal{A}}
\newcommand{\items}{\mathcal{O}}

\newcommand{\mw}[1]{\textcolor{blue}{\textit{#1}}}

\newcommand{\jl}[1]{\textcolor{red}{\textit{#1}}}
\if01
\usepackage[style=alphabetic, isbn=false, doi=false, url=false, backend=biber, maxbibnames=99]{biblatex}
\renewbibmacro{in:}{%
  \ifentrytype{article}{}{%
  \printtext{\bibstring{in}\intitlepunct}}}
\DeclareFieldFormat[article,incollection]{pages}{#1}%
\fi

\if01
\usepackage[style=alphabetic, isbn=false, doi=false, url=false, backend=bibtex, maxbibnames=99]{biblatex}

\fi
\usepackage{url}
\usepackage{hyperref}
\usepackage{fourier}

\renewcommand{\emph}[1]{\textbf{\textit{#1}}}

\title{New algorithms for matching problems}
\author{Jacky Lo and Mark C. Wilson}

\begin{document}

\maketitle

\begin{abstract}
The standard two-sided and one-sided matching problems, and the closely related school choice problem, have been widely studied from an axiomatic viewpoint. A small number of algorithms dominate the literature. For two-sided matching, the Gale-Shapley algorithm; for one-sided matching, (random) Serial Dictatorship and Probabilistic Serial rule; for school choice, Gale-Shapley and the Boston mechanisms. 

The main reason for the dominance of these algorithms is their good (worst-case) axiomatic behaviour with respect to notions of efficiency and strategyproofness. However if we shift the focus to fairness, social welfare, tradeoffs between incompatible axioms, and average-case analysis, it is far from clear that these algorithms are optimal.

We investigate new algorithms several of which have not appeared (to our knowledge) in the literature before. We give a unified presentation in which algorithms for 2-sided matching yield 1-sided matching algorithms in a systematic way. In addition to axiomatic properties, we investigate agent welfare using both theoretical and computational approaches. We find that some of the new algorithms are worthy of consideration for certain applications. In particular, when considering welfare under truthful preferences, some of the new algorithms outperform the classic ones.

\end{abstract}



\section{Introduction}

The standard two-sided and one-sided matching problems, and the closely related school choice problem, have been widely studied from an axiomatic viewpoint. A small number of algorithms dominate the literature. For two-sided matching, the Gale-Shapley algorithm; for one-sided matching, (random) Serial Dictatorship and Probabilistic Serial rule; for school choice, Gale-Shapley and the Boston mechanisms. 

The main reason for the dominance of these algorithms is their good axiomatic behaviour with respect to notions of efficiency and strategyproofness. However if we shift the focus to fairness, social welfare, or tradeoffs between incompatible axioms, it is far from clear that these algorithms are optimal.

\subsection{Our contribution}
\label{s:contrib}

In Section~\ref{s:algo} we introduce  several (in our opinion) natural algorithms for one-sided matching, several of which have not appeared (to our knowledge) in the literature before. We give a consistent derivation using specializations of the Gale-Shapley algorithm \cite{} for two-sided matching, which includes the well-known algorithms Serial Dictatorship \cite{} and Naive Boston \cite{} in a unified framework. In Section~\ref{s:props}, in addition to axiomatic properties such as efficiency and strategyproofness, we investigate welfare loss using a computational approach. We find that under truthful preferences, some of the new algorithms clearly outperform the classic ones. In particular, we recommend some new algorithms for some applications.

\section{Definitions and terminology}

Let $\agents = \{a_1, \dots , a_n\}$ be a finite set of \emph{agents} and $\items = \{o_1, \dots, o_m\}$ 
a finite set of \emph{items}. 

In general, the number of items and agents may not be equal. We focus on the case $m=n$ in the present article. The more general case involves substantial complications: different ways of assigning preferences over subsets of $\items$ or $\agents$ lead to different notions of strategyproofness, for example. However all our algorithms can be modified trivially in order to work in the general case.

In the standard two-sided matching problem, each element of $\items$  has a complete strict preference order for elements of $\agents$, and vice versa, while for $1$-sided matching only the latter information is required. For school choice the order of preference of items over agents depends on the preferences of agents over items --- the two sides are definitely not independent (schools typically must admit students who are qualified as long as there is capacity, and use their own preferences only when a tie must be broken). We aim to unify these three cases, and restrict to the case of strict linear orders.

Let $L(\items)$ (respectively $L(\agents)$) denote the set of all strict linear orders on $\items$ (resp. agents).  A preference \emph{profile} is a pair of functions $(\pi_A, \pi_I)$ where  $\pi_A: A\to L(\items)$ and $\pi_I: \items \to L(\agents)$. A  \emph{matching} or \emph{discrete assignment} is a function $f:\items \leftrightarrow \agents$. Let $S$ be the set of all doubly stochastic $n\times n$ matrices with rows indexed by agents and columns by items.  A \emph{random assignment} is an element of $S$. The \emph{matching problem} is simply to output a matching given an input profile. The \emph{proportional} assignment is the random assignment in which each matrix entry equals $1/n$.
 
Of course, there are $n!$ discrete assignments and finding one is trivial. The point is to find one with desirable properties. A discrete assignment is \emph{efficient} if there is no other assignment which improves the outcome for some agent and does not worsen it for any agent. An algorithm for randomized assignment is \emph{ex-post efficient} if every matching occurring with positive probability is an efficient assignment. 

We now review some standard algorithms from the literature. 

Two algorithms for the same matching problem are \emph{equivalent} if they each produce the same output for every input. All the algorithms under study are \emph{anonymous}, meaning that a permutation of the players leads to the same permutation of the assignment. In other words, only the preferences matter, not the agents' identities. When discrete assignments are used, this means that whenever we have two agents with identical preferences, one must envy the other's assignment. A stronger condition is \emph{symmetry} (also called ``equal treatment of equals'') which says that agents with the same preferences receive the same assignment. Clearly, this can only be satisfied in the framework of random assignments.

Every algorithm that uses a fixed initial order of agents and produces a matching can yield an algorithm that produces a random matching, simply by randomizing over the initial order. This is usually done according to the uniform distribution, in order to preserve symmetry between agents. Thus every algorithm discussed in Section~\ref{s:algo} has a randomized version, which we denote by prefixing ``R" to its name.

\subsection{One-Sided Algorithms}
\label{ss:1-sided}

One-sided matching refers to the situation where the items' preferences over agents are ignored.

The most commonly discussed algorithm is Serial Dictatorship \cite{}.

\begin{eg} (Serial Dictatorship)
\label{eg:SD}

Fix an arbitrary linear ordering on $A$. The \emph{Serial Dictatorship} algorithm (SD) with respect to this ordering assigns items to agents as follows: at step $i$ , allocate to agent $i$th its most preferred item that has not already been allocated to a previous agent.
\end{eg}

The randomized version is denoted RSD, as mentioned above.
Serial Dictatorship satisfies important axiomatic properties such as \emph{ex-post efficiency} and \emph{strategyproofness} (all axiomatic properties are defined and discussed in Section~\ref{s:props}). 

The next algorithm was proposed by Bogolmanaia and Moulin \cite{BoMo2001}. It is \emph{ordinally efficient} in addition to being ex-post efficient, and has a weak strategyproofness property. It can be described using a cake-eating analogy.

\begin{eg} (Probabilistic Serial)
\label{eg:PS}
The \emph{Probabilistic Serial} rule is inherently randomized, and generates a random assignment as follows. We interpret an assignment of fraction $x$ of item $j$ to agent $i$ to mean that $i$ receives fraction $x$ of $j$ (in other words we pretend that the items are infinitely divisible). All agents simultaneously begin ``eating" at unit speed, each agent at each instant eating from its most preferred item among those that have not been completely consumed. On termination we have a random assignment.
\end{eg}

A closely related problem to one-sided matching is the \emph{housing market} problem \cite{ShSc1974}. The difference is that every agent is assumed to have an initial allocated item, and we seek a method for finding an allocation that is optimal in some way.

\begin{eg} (Top Trading Cycle)
\label{eg:GTTC}
When each agent is considered to be initially assigned an entire item, the agents may trade amongst themselves as follows. Each agent $i$ points to the agent currently owning the item on the top of $i$'s preference list. By finiteness and since every node has outdegree 1, this directed graph must contain a cycle. Reallocate items according to the arcs in the cycle, and remove these agents and items from further consideration. Repeat (using pointers to the next level preference if necessary) until no items/agents remain. This \emph{TTC algorithm} \cite{ShSc1974}, attributed by Shapley and Scarf to David Gale, always yields a discrete assignment that is efficient, and the mechanism is strategyproof and individually rational. Furthermore the algorithm runs in polynomial time.

\end{eg}

\begin{eg} (sample execution of TTC) 
Consider the profile where agents $1$ and $2$ have preferences  $a > b > c$ and agent $3$ has preference  $b > a > c$. The assignment $1:c, 2:b, 3:a$ is not ex-post efficient, because $2$ and $3$ can trade to their mutual benefit. In the first round of TTC, agents $1$ and $2$ points to agent $3$, while agent $3$ points to agent $2$. There is a cycle between agents $2$ and $3$. The agents swap along the cycle and are removed from consideration. Agent $1$ then points to itself in the next round, swaps along the cycle, remains with item $c$, and is removed from consideration. With no agents left, the TTC algorithm halts. The output is the ex-post efficient assignment $1:c, 2:a, 3:b$.
\end{eg}

When using TTC we have freedom in the choice of initial assignment. For example, choosing this uniformly at random and running TTC yields an algorithm equivalent to RSD \cite{AbSo1998}, and an adaptation of TTC to trade unit shares yields an algorithm equivalent to PS when run on the proportional endowment \cite{Kest2009}.

We find it useful to run TTC on the output of some of our algorithms in Section~\ref{s:algo} (algorithms which satisfy ex-post efficiency gain no benefit from running TTC, which terminates immediately because there is no cycle). The resulting combined algorithms, denoted XG where X is the name of the basic algorithm, are ex-post efficient and appear to have considerably better overall performance than the original algorithms.

\subsection{Two-sided algorithms}
\label{ss:2-sided}
In this case agents have complete strict preferences over items, and vice versa. The most well-known algorithm belongs to the class of \emph{deferred acceptance} algorithms. Items and agents are tentatively matched, but these ``engagements'' may be broken. In fact each item may attach to up to $n$ agents in the course of the algorithm. 

In the \emph{Gale-Shapley} algorithm \cite{GaSh1962}, agents in turn approach previously unapproached items that they prefer to their currently assigned item (every agent prefers each item to not having an item, and every item prefers every agent to not being held by an agent). If the currently proposing agent is preferable to the agent currently matched with the item, the item will reject its current partner for the proposing agent. No agent may approach an item that has already rejected it (such an approach would lead to another rejection by the above rules). This ensures termination after at most $n^2$ proposals.

\begin{eg}
\label{eg:GS}
This is adapted from Example 2 in \cite{GaSh1962}.

Suppose that the proposers' preferences are as follows:
\begin{align*}
1: \quad a>b>c>d \\
2: \quad a>d>c>b \\
3: \quad b>a>c>d \\
4: \quad d>b>c>a \\
\end{align*}

and the proposees' preferences are given by

\begin{align*}
a: \quad 4>3>1>2 \\
b: \quad 2>4>1>3 \\
c: \quad 4>1>2>3 \\ 
d: \quad 3>2>1>4 \\
\end{align*}

\if01
The sequences of proposals are as follows:

\begin{tabular}[!ht]{|c|c|c|}
\hline
	Proposal & Outcome & Current Partial Matching\\
\hline
    $a_1 =>$ a & tentatively matched & $a_1$:a\\
    $a_2 =>$ a & $a_2$ rejected & $a_1$:a\\
    $a_3 =>$ b & tentatively matched & $a_1$:a, $a_3$:b\\
    $a_4 =>$ d & tentatively matched & $a_1$:a, $a_3$:b, $a_4$:d\\
    $a_2 =>$ d & $a_4$ rejected & $a_1$:a, $a_2$:d, $a_3$:b\\
    $a_4 =>$ b & $a_3$ rejected & $a_1$:a, $a_2$:d, $a_4$:b\\
    $a_3 =>$ a & $a_1$ rejected & $a_2$:d, $a_3$:a, $a_4$:b\\
    $a_1 =>$ b & $a_1$ rejected & $a_2$:d, $a_3$:a, $a_4$:b\\
    $a_1 =>$ c & tentatively matched & $a_1$:c, $a_2$:d, $a_3$:a, $a_4$:b\\
\hline
\end{tabular}
\fi

The final matching is $1:c, 2:d, 3:a, 4:b$. There are 9 proposals made during the execution of the algorithm. 
\end{eg}

The Gale-Shapley algorithm has the well-known property that the output matching is \emph{stable}, meaning that there is no unmatched (agent, item) pair who each prefer each other to their current partner. Also, the output matching is optimal for proposers, meaning that each proposer receives the best possible item it can receive in a stable matching. It follows that the output of the algorithm does not depend on the order of proposals made by agents. Note that, by contrast, the output of the algorithms in Section~\ref{s:algo} will depend strongly on the order of proposals.

We can produce one-sided matching algorithms from two-sided ones by forcing the items to have specific (fictitious) preferences. We use this idea systematically in Section~\ref{s:algo}. 

\subsection{School choice algorithms}
\label{ss:school}

A case intermediate between 1-sided and 2-sided matching, which is important for later, is that of \emph{school choice}. The Gale-Shapley algorithm is applicable to the case where the number of agents exceeds the number of items, provided items (schools) have capacity for some number of agents (students). A school accepts a student's proposal provisionally provided there is capacity remaining, or the student is preferable to an already tentatively accepted student. The case where each school has capacity $1$ and the numbers of schools and students are equal is the case described in Section~\ref{ss:2-sided}.

There are other algorithms for school choice that use \emph{immediate} acceptance. In this case, each student first applies to her first choice school. Each school ranks applicants and chooses as many as it can, subject to capacity. Students not accepted already then apply to their second choice school, etc. This description implicitly uses simultaneous proposing by all unmatched agents and is called the \emph{Boston mechanism} \cite{AbSo2003, MeSe2014}. There is also a sequential version in which proposals are made one agent at a time. In that case, the order of proposals clearly changes the final allocation, since no engagement can ever be broken.

The special case of the Boston mechanism in which each school has capacity $1$ and there are equal numbers of agents and items is a 2-sided matching algorithm as defined above. For the same input as in Example~\ref{eg:GS}, the final allocation using sequential offers by $1,2,3,4$ in that order is $1:a, 2:d, 3:b, 4:c$. The final allocation using simultaneous offers is $1:a, 2:c, 3:b, 4:d$.

\section{New algorithms}
\label{s:algo}

For the rest of the analysis, we construct one-sided matching algorithms by relaxing 2-sided algorithms. The Gale-Shapley algorithm generates a stable matching using a series of proposals, based on fixed preferences of the agents and items. Given the agent preferences over items as input to a 1-sided matching problem, we construct fictitious preferences for the items over the agents. For example, we can assume that all items have a fixed common preference.

\begin{prop} SD is equivalent to a special case of GS.
\label{prop:GS yields SD}
\end{prop}
\begin{proof}
Suppose that all items have the same preference order over agents, which without loss of generality we write $1>2>\dots>n$. We claim that Gale-Shapley will output an assignment that is the same as the output of Serial Dictatorship with the agent order $1, 2, \dots, n$.

The proof is inductive. The base case is that agent $1$ will get his first choice with GS. It is trivially true as agent $1$ will propose to its most preferred item, and since every item prefers agent $1$ to any other agent, they cannot be rejected later. Therefore agent $1$ will be allocated the same item under GS or SD.

Stability of GS implies that for every item that agent $i$ wants more than the item they are allocated, the item must be held by an agent ranked higher by the item. If every agent before $i$ gets its choice as per SD, agent $i$ will eventually propose to its choice under SD. As every agent after $i$ is ranked below $i$ by all items, they cannot cause that item to reject $i$. Therefore agent $i$ will have the same item under GS or SD.
\end{proof}

We can do the same thing with the (simultaneous) Naive Boston algorithm. Given an instance of 1-sided matching, we create fictitious preferences in which each item has the same preference, say $1>2>\dots >n$. The resulting algorithm we call the $1$-sided Naive Boston algorithm. Note that we can also interpret this algorithm sequentially if we ensure that agent order is $1,2,\dots, n$, but otherwise the sequential and simultaneous forms will differ in general.

Below, we generalize this fictitious preference approach by allowing each item to build its fictitious preferences dynamically, using some fixed rule, based only on the sequence of proposals that it receives from agents. Recall that any order of proposals gives the same result for the Gale-Shapley algorithm with fixed preferences, but as we see below, this is not the case in our relaxed setup. 

\begin{defn}Throughout the rest of this article, for the purposes of illustration and comparison between algorithms we use what we call the \emph{standard profile} in which agents $1, 2$ and $3$ have preferences $a > b > c > d$, and agent $4$ has preferences $b > a > c > d$. 
\end{defn}

\subsection{The ``permanent memory" case}
\label{ss:memory}

As the order of the agents' proposals affects the items' preferences, the order of proposals affects the final allocation. This has two consequences. The first is that the treatment of rejected agents matters. After an agent is rejected, either because the item prefers its current agent or breaks its tentative engagement, it may not necessarily be the next agent to propose. We consider two possibilities:  a stack (rejected agents go to the top of the stack) or a queue (rejected agents go to the back of the queue).  The other consequence is, as noted above, that the initial order of agents has an impact on the final allocation. For the purpose of the analysis, the algorithms will fix an arbitrary initial order.

We consider two rules for building preferences dynamically. These are \emph{early-proposal preference} (or \emph{Accept-First}) and \emph{late-proposal preference} (or \emph{Accept-Last}). Accept-First means that the first agent to approach an item is accepted, and subsequent proposals are rejected. Using Accept-Last, an item always breaks an engagement in favour of a new proposer, if it has not yet been held by the new proposer. In terms of the marriage interpretation often used to describe the Gale-Shapley algorithm, for Accept-First algorithms the proposees stick faithfully to their first suitor, whereas for Accept-Last algorithms the proposees are always more satisfied with a new suitor than their current fianc\'{e}.

These two dichotomies (stack/queue, Accept-First/Accept-Last) when combined with the distinction between permanent and temporary memory (explained in Section~\ref{ss:no memory}) yield 8 algorithms. We denote them by three-letter abbreviations. For example, PFS refers to permanent memory, Accept-First, stack.

Sample executions on the standard profile of all 8 algorithms introduced below can be found in Appendix~\ref{apps:examples}. The results are summarized in Table~\ref{t:stdprof}.

We consider the Accept-First algorithms. The algorithm PFS is equivalent to Serial Dictatorship. Interestingly, switching the stack to a queue results in an algorithm that is equivalent to the Naive Boston algorithm, so we obtain no new inequivalent algorithms in this case, just a unified presentation of old ones. We give the details below.

\begin{prop} PFS with order of agents $1, 2, \dots, n$ is equivalent to Serial Dictatorship with the same order of agents.
\end{prop}
\begin{proof}
By definition PFS is a special case of GS in which there is a common preference order $1>2>\dots > n$ for items over agents, because engagements are never broken (by Accept-First and the permanent memory). By Proposition~\ref{prop:GS yields SD} the latter is equivalent to SD with the agent order $1,2, \dots, n$.
\end{proof}

\begin{prop} PFQ with order of agents $1,2,\dots ,n$ is equivalent to the $1$-sided Naive Boston algorithm where the common preference order of items is $1>2>\dots >n$.
\end{prop}
\begin{proof}
It is useful to consider the above algorithm as occurring in rounds. Round $i$ ends precisely when all the remaining agents have proposed to their $i$th choice. We show by induction that during round $i$:
\begin{itemize}
\item the agents proposing are precisely those who have not been matched previously;
\item each such agent makes exactly one proposal, to its $i$th choice.
\end{itemize}
Since this is exactly the behaviour of the (simultaneous) 1-sided Boston algorithm and since assignments are never changed in either algorithm, the proposition follows.

When $i=1$, all agents have proposed to their first choice. Thus each agent has made exactly one proposal, because any agent that is accepted never makes another proposal, and any agent that is rejected must go to the back of the queue and wait until all other agents have made their first proposal. 

Assuming the result holds for all rounds before $i$, then all remaining agents have proposed to and been rejected by all items down to rank $i-1$. In round $i$ they must then propose to their $i$th choice. This happens exactly once because of the queue discipline.
\end{proof}

\begin{eg} (Accept-Last permanent memory algorithms versus accept-first)

For the standard profile, the final assignment using PFS is $1:a, 2:b, 3:c, 4:d$, while using PLS it is $1:d, 2:c, 3:a, 4:b$. If we change the stack to a queue, the final assignment is $1:a, 2:c, 3:d, 4:b$ for PFQ and $1:d, 2:c, 3:b, 4:a$ for PLQ. 
Note that with these preferences, some agent must receive its 4th choice and no more than two agents can receive their 1st choice. Only PFQ achieves the latter condition.
\end{eg}

Although the Accept-Last algorithms are not ex-post efficient, empirical results show that (among other things) their output, when used as an initial endowment for TTC, leads to better welfare performance than the ex-post efficient Accept-First algorithms. We discuss this in details in Section~\ref{s:props}. 

\subsection{The ``temporary memory" case}
\label{ss:no memory}

The previous algorithms assume that each item retains its preferences throughout the execution of the algorithms. Other interesting algorithms can be generated by relaxing that requirement. Whenever the item resets its memory, it make sense for the agents to propose to items that rejected them before. Any rules that uses temporary memory for items must ensure that the algorithm will halt with a matching. To ensure the algorithm halts, we only allow the item to reset its memory when the number of tentative matchings has increased. As items do not go from matched to unmatched, this happens precisely when a new item is matched.

Whenever an agent proposes to an unmatched item, all items, including the new item, lose their memory of preferences. When an agent proposes to a matched item with no preferences, the item prefers the proposing agent instead of the matched agent. As the number of tentative matchings increases  throughout the execution of the algorithm, there can only be $n$ resets of the preferences, and thus the number of proposals is bounded by $n^3$.

Using the above rule, four new algorithms analogous to those in Section~\ref{ss:memory} can be constructed. 

We first present the Accept-First algorithms. The temporary memory analogue of Serial Dictatorship, namely TFS, is interesting. Each round of this algorithm operates like a Serial Dictatorship in which a subset of the agents repeatedly ``steal" items in chains until some agent chooses an unmatched item, whereupon a new round begins with a new agent beginning the stealing. The order of agent choices in a round is not fixed, but determined by the stealing process. We suggest the alternative name ``Iterative Dictatorship'' for this algorithm. The algorithm TFQ is the temporary memory analogue of the $1$-sided Boston algorithm. Like all queue-based algorithms it is harder to interpret than a stack-based algorithm.

The Accept-Last algorithms are harder to understand (but see the party interpretation below, which was the inspiration for our entire research program).

\begin{eg} (Temporary memory, Accept-First versus Accept-Last)
For the standard profile, the final assignment under TFS is $1:d, 2:a, 3:c, 4:b$ and the final assignment using TFQ is $1:a, 2:b, 3:d, 4:c$. By contrast, the final assignment using TLS is $1:b, 2:a, 3:d, 4:c$ while the final assignment under TLQ is $1:a, 2:b, 3:d, 4:c$. 
\end{eg}

\subsection{Further comments}
\label{ss:further}
We have presented 8 algorithms in a unified framework, corresponding to the dichotomies memory/no memory, stack/queue, Accept First/Accept Last. Basic description of their behaviour on the standard profile is shown in Table~\ref{t:stdprof}. Note that all give different outputs on this input, except TFQ and TLQ, which are of course different in general. A stronger statement, namely that all randomized versions are inequivalent algorithms, is shown by example in Appendix~\ref{app:alldiff}. Also note that on this input, only PLQ fails to give an efficient allocation. Running TTC on the output of PLQ yields the allocation $1:a, 2:c, 3:b, 4:d$.

\begin{table}
\caption{Behaviour of algorithms on standard profile where $n=4$}
\label{t:stdprof}
\begin{tabular}{ccc}
\hline
Algorithm & Output matching & Number of proposals \\
\hline
PFS & 1:a, 2:b, 3:c, 4:d & 10\\
PFQ & 1:a, 2:c, 3:d, 4:b & 9\\
PLS & 1:d, 2:c, 3:a, 4:b & 9\\
PLQ & 1:d, 2:c, 3:b, 4:a & 10\\
TFS & 1:d, 2:a, 3:c, 4:b & 18\\
TFQ & 1:a, 2:b, 3:d, 4:c & 33\\
TLS & 1:b, 2:a, 3:d, 4:c & 18\\
TLQ & 1:a, 2:b, 3:d, 4:c & 21\\
\end{tabular}

\end{table}

All algorithms can be interpreted in terms of a party game, with the host providing a stash of presents. As each person arrives at the party (say though a narrow door), they take a present from the stash or (in some cases) from another person. Permanent memory algorithms have a single round, and temporary memory algorithms begin a new round every time a new present is taken. For Accept-First algorithms, in each round each present can be taken at most once. 
For Accept-Last algorithms, in each round each (person, present) pair can occur at most once. The queue or stack discipline determines what happens to a partygoer when it loses its present: choose a replacement present immediately, or go to the back of the queue. The Accept-First queue-based algorithms would be uninteresting, as would PFS, but the others seem to us worth trying.

The TLS algorithm is closely related in this interpretation to the party game \emph{Yankee Swap} or White Elephant \cite{wiki:yankee}, which was the inspiration for our research program. In the real game, presents are contributed by partygoers and are wrapped, so no person has full information on their own preference. We are not aware of any other real-life party games based on the other algorithms.

\section{Properties of the algorithms}
\label{s:props}

Most of our algorithms fail to satisfy any of the common axiomatic properties. However, some have good average-case behaviour. Interestingly, they behave quite differently from each other. 

\subsection{Ex-post Efficiency}
\label{ss:expost}

\begin{prop} All Accept-First algorithms are ex-post efficient.
\end{prop}
\begin{proof}
We show by induction on the round number that the partial allocation constructed so far is ex-post efficient (in each case a round ends every time a previously unmatched item is chosen --- note that this is a different usage of ``round" to that in Section~\ref{ss:further}, where permanent memory algorithms have a single round). The first round always terminates with the first agent taking its top choice, and this is obviously an ex-post efficient outcome. Suppose that the result holds for all rounds before $i$ and consider round $i$. The entering agent chooses an item and retains it throughout the round (by Accept-First policy). Each agent taking an unmatched item or stealing an item during this round (an ``active agent") chooses the best item available. Such an agent $j$ would only wish to trade with an agent $k$ who has chosen since the last memory reset (in the permanent memory case, since the beginning of the algorithm; in the temporary memory case, since the beginning of the round). However in this case $k$ will not wish to trade with $j$. Thus there can be no mutually beneficially trading cycle within the group of active agents. By inductive hypothesis there is also no such cycle within the group of inactive agents. Each agent in the active group prefers its current item to everything held by the inactive group because such items were available to steal. The result follows because TTC will terminate with no trades.
\end{proof}

\begin{eg} All Accept-Last algorithms fail ex-post efficiency. To see this, consider the profile where agents 1,2,3 have respective preference orders $a>b>c, a>b>c, b>a>c$. Direct computation shows that for each algorithm X, there is some initial agent order such that algorithm XG formed by running TTC on the output of X gives a different result. Thus X cannot be ex-post efficient. 
\if01
(PM AL S) - agent 1 \& 2 a>b>c, agent 3 b>a>c, agent 1:a, agent 2:a, agent 1:b agent 3:a agent 2:b agent 1:c, 1:c 2:b 3:a, agent 2 and 3 can swap
(PM AL Q) - \jl{same preferences as MLS
Memory AcceptLast Queue (in Pref Order):
0.000000 0.500000 0.500000 
0.000000 0.500000 0.500000 
0.000000 1.000000 0.000000 
Memory AcceptLast Queue +GTTC(in Pref Order):
0.500000 0.000000 0.500000 
0.500000 0.000000 0.500000 
1.000000 0.000000 0.000000 }
(TM AF S) - with 1 agent and m item, the allocation is ExEff (trivial). Assuming that i-1 agents and m item allocation is ExEff, then i agent and n item also is. Any agents that get changes his allocated item when the $i^th$ agent comes in has the same item as if the algorithm is a SD, thus those agents cannot improve on his allocation. Agents that does not change his item cannot trade with each other because the allocation with i-1 agents are ExEff. Thus at n agents and m item, the allocation is ExEff.
\fi
\if01

details

Memory AcceptLast Stack (in Pref Order):
0.166667 0.333333 0.500000 
0.166667 0.333333 0.500000 
0.333333 0.666667 0.000000 
Memory AcceptLast Stack +GTTC(in Pref Order):
0.500000 0.000000 0.500000 
0.500000 0.000000 0.500000 
1.000000 0.000000 0.000000 
NoMemory AcceptLast Stack (in Pref Order):
0.000000 0.500000 0.500000 
0.000000 0.500000 0.500000 
0.000000 1.000000 0.000000 
NoMemory AcceptLast Stack +GTTC(in Pref Order):
0.500000 0.000000 0.500000 
0.500000 0.000000 0.500000 
1.000000 0.000000 0.000000 
Memory AcceptLast Queue (in Pref Order):
0.000000 0.500000 0.500000 
0.000000 0.500000 0.500000 
0.000000 1.000000 0.000000 
Memory AcceptLast Queue +GTTC(in Pref Order):
0.500000 0.000000 0.500000 
0.500000 0.000000 0.500000 
1.000000 0.000000 0.000000 
NoMemory AcceptLast Queue (in Pref Order):
0.000000 0.500000 0.500000 
0.000000 0.500000 0.500000 
0.000000 1.000000 0.000000 
NoMemory AcceptLast Queue +GTTC(in Pref Order):
0.500000 0.000000 0.500000 
0.500000 0.000000 0.500000 
1.000000 0.000000 0.000000 
\fi

\end{eg}

\subsection{Ordinal Efficiency}
\label{ss:ordinal}

\if01 details 
\jl{The NM L S/Q algorithms are almost o-eff, with 98\% and 96\% IANC preferences n=4 are o-eff. counter examples:
Preference
a>b>c>d
a>b>c>d
b>a>d>c
b>c>a>d
NoMemory AcceptLast Stack +GTTC(in Pref Order):
0.208333 0.083333 0.416667 0.291667 
0.208333 0.083333 0.416667 0.291667 
0.000000 0.583333 0.416667 0.000000 
0.833333 0.166667 0.000000 0.000000 
NoMemory AcceptLast Queue +GTTC(in Pref Order):
0.250000 0.250000 0.250000 0.250000 
0.250000 0.250000 0.250000 0.250000 
0.000000 0.500000 0.500000 0.000000 
0.500000 0.500000 0.000000 0.000000 
}
\fi

A random allocation $S$ is ordinally efficient if there is not another random allocation, $S'$ that each agent $i$ SD-prefers $S'_i$ to $S_i$, with at least 1 agent strictly SD-preferring their allocation under $S'$. Ordinal efficiency implies ex-post efficiency, but not vice versa \cite{BoMo2001}.

\begin{eg} All Accept-First algorithms fail ordinal efficiency. A counterexample (details omitted): suppose that agents 1, 2 prefer $a>b>c>d$, while agents 3, 4 have preference $a>b>d>c$. This same counterexample also works for the ex-post efficient algorithms PLSG and PLQG. A different counterexample with $n=4$ works for TLSG and TLQG, and ordinal efficiency seems to be violated less often for these two algorithms.

\if01 details
\jl{
The two memory accept last algorithm also fails o-eff on this preference.
Memory AcceptFirst Stack (in Pref Order):
0.250000 0.250000 0.416667 0.083333 
0.250000 0.250000 0.416667 0.083333 
0.250000 0.250000 0.416667 0.083333 
0.250000 0.250000 0.416667 0.083333 
NoMemory AcceptFirst Stack (in Pref Order):
0.250000 0.250000 0.416667 0.083333 
0.250000 0.250000 0.416667 0.083333 
0.250000 0.250000 0.416667 0.083333 
0.250000 0.250000 0.416667 0.083333 
Memory AcceptLast Stack (in Pref Order):
0.250000 0.250000 0.416667 0.083333 
0.250000 0.250000 0.416667 0.083333 
0.250000 0.250000 0.416667 0.083333 
0.250000 0.250000 0.416667 0.083333 
Memory AcceptLast Stack +GTTC(in Pref Order):
0.250000 0.250000 0.416667 0.083333 
0.250000 0.250000 0.416667 0.083333 
0.250000 0.250000 0.416667 0.083333 
0.250000 0.250000 0.416667 0.083333 
Memory AcceptFirst Queue (in Pref Order):
0.250000 0.250000 0.416667 0.083333 
0.250000 0.250000 0.416667 0.083333 
0.250000 0.250000 0.416667 0.083333 
0.250000 0.250000 0.416667 0.083333 
NoMemory AcceptFirst Queue (in Pref Order):
0.250000 0.250000 0.416667 0.083333 
0.250000 0.250000 0.416667 0.083333 
0.250000 0.250000 0.416667 0.083333 
0.250000 0.250000 0.416667 0.083333 
Memory AcceptLast Queue (in Pref Order):
0.250000 0.250000 0.416667 0.083333 
0.250000 0.250000 0.416667 0.083333 
0.250000 0.250000 0.416667 0.083333 
0.250000 0.250000 0.416667 0.083333 
Memory AcceptLast Queue +GTTC(in Pref Order):
0.250000 0.250000 0.416667 0.083333 
0.250000 0.250000 0.416667 0.083333 
0.250000 0.250000 0.416667 0.083333 
0.250000 0.250000 0.416667 0.083333 
}
\fi
\end{eg}

\subsection{Strategyproofness}
\label{ss:strategyproof}

An assignment algorithm is \emph{strategyproof} if no agent has incentive to misreport its preferences, irrespective of what other agents do. In other words, truthful reporting is a strictly dominant strategy for each agent. This condition is of course rather strong, but is known to be satisfied by Serial Dictatorship and hence by PFS. All other algorithms discussed here fail to satisfy it as we show below. For random allocation algorithms, the incentive is understood to be expressed in terms of expected utility in the usual way.

The algorithm Probabilistic Serial satisfies \emph{weak strategyproofness} \cite{BoMo2001}, which says that no agent has incentive  to deviate, where the incentive is expressed in terms of first-order stochastic dominance. In other words, for each profile and each agent, there is \emph{some} consistent choice of utility function for that agent for which deviation from truthfulness is unprofitable. Note that this is weaker than strategyproofness, which says that for each profile and agent, and for \emph{every} consistent choice of utility function for that agent, deviation from truthfulness is unprofitable.

Clearly, if algorithm X is strategyproof for every initial order of agents, so is RX. 

We show by examples that none of the randomized versions of our algorithms are weakly strategyproof. Consider the case where agents $1,2,3,4$ all have preference $a>b>c>d$. 
If some agent submits instead $a>c>b>d$ while the others remain truthful, 
algorithms TLS, TLG, PLS, PLQ, TLQ, TLQG obtain a preferable outcome for that agent. 

Similar examples show that PLQG, PLSG, TFS and TFQ, and PFQ all fail weak strategyproofness.  

\if01 
details
NMALS: yields 3/4 1/4 0 0
NMALS+G: 3/4 0 1/4 0
MALS: 1/4 1/2 1/4 0
MALQ: 1/4 3/4 0 0
NMALQ: 3/4 1/4 0 0
NMALQ+G: 3/4 0 1/4 0

With MALQ+G:
True Preferences:
a,b,c,d
a,b,c,d
a,d,b,c
a,c,d,b
Agent 4 gets 1/4 0 7/12 1/6
By submitting a,c,b,d Agent 4 gets 1/4 0 3/4 0

NMAFS \& NMAFQ:
True Preferences:
a,b,c,d
a,b,c,d
a,b,c,d
a,c,b,d
NMAFS: agent 4 gets 1/4 0 1/4 1/2
NMAFQ: agent 4 gets 0 1/4 0 3/4

It is trivial to see that if agent 4 submits a,b,c,d, they will get the same proportional allocation, which SD dominates the allocation with true preferences.

NB:
True preferences:
a,b,c,d
a,b,c,d
c,a,b,d
a,c,b,d

Agent 4 gets 1/3 0 0 2/3
By submitting a,b,c,d, they will get 1/3 1/3 0 1/3 which SD dominates the allocation from true preference
\fi

\if01
\jl{For MALS+G, agent 1-3 pref = a>b>c>d, agent 4 pref = b>c>d>a. allocation is 
Memory AcceptLast Stack +GTTC(in Pref Order):
0.333333 0.250000 0.250000 0.166667 
0.333333 0.250000 0.250000 0.166667 
0.333333 0.250000 0.250000 0.166667 
0.250000 0.250000 0.500000 0.000000 
if agent 4 submits b>c>a>d, allocation is
Memory AcceptLast Stack +GTTC(in Pref Order):
0.333333 0.083333 0.250000 0.333333 
0.333333 0.083333 0.250000 0.333333 
0.333333 0.083333 0.250000 0.333333 
0.750000 0.250000 0.000000 0.000000 }
\fi

\subsection{Utilitarian welfare}
\label{ss:welfare}

We give a basic analysis here, and refer the reader to our more extensive analysis \cite{WiLo2017b}. We analyse here only the average-case performance of the randomized versions of the algorithms under truthful behaviour. We impute utility values of agents by requiring them to all have the same (Borda) utility function, whereby the $i$th choice corresponds to utility $n-1+i$. We consider the \emph{utilitarian social welfare}, which is the sum over all agents of the utility of their allocation. The optimal value of the utilitarian social welfare is efficiently computable, for example using the \emph{Hungarian algorithm} \cite{}. This allows us to quantify the fraction of the maximum possible social welfare that is lost, on average, by each algorithm. We use a Java implementation by K.~Stern \cite{Ster2012}.

Results are shown in Figure~\ref{fig:util_8}. They show that RSD is outperformed substantially by PFQ, TFQ and TFS, but is better than the Accept-Last algorithms. In Figure~\ref{fig:util_TTC} we show how the welfare improves substantially when TTC is run on the output of our Accept-Last algorithms. In Figure~\ref{fig:util_5} we show how our best-performing algorithms, namely TLQG and TLQS, compare with the standard algorithms RSD, PS and Naive Boston. In fact all the Accept-Last algorithms with TTC outperform those standard algorithms.

\begin{figure}
\caption{Utilitarian welfare loss of our 8 basic algorithms}
\label{fig:util_8}
\includegraphics[width=14cm]{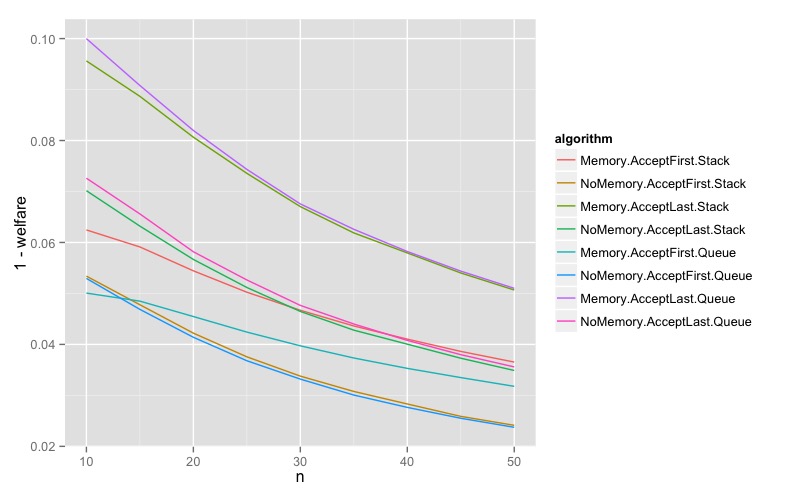}
\end{figure}

\begin{figure}
\caption{Utilitarian welfare improvement for Accept-Last with TTC}
\label{fig:util_TTC}
\includegraphics[width=14cm]{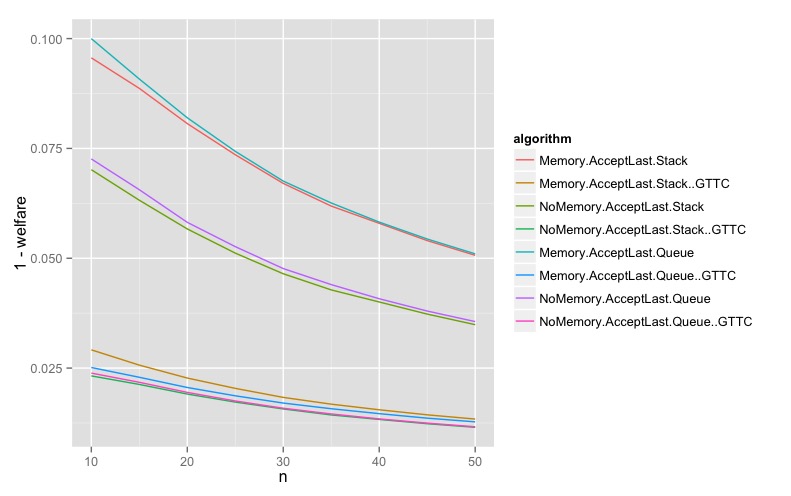}
\end{figure}

\begin{figure}
\caption{Utilitarian welfare loss comparison}
\label{fig:util_5}
\includegraphics[width=14cm]{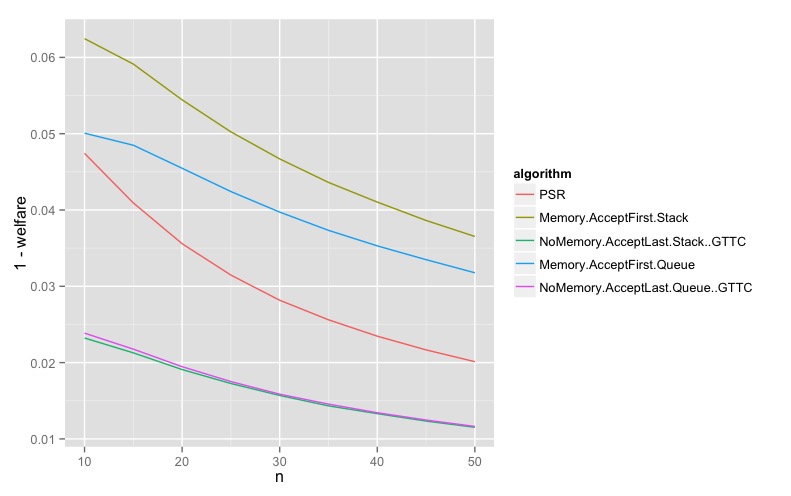}
\end{figure}

\subsection{Egalitarian welfare}
\label{ss:egal}

We also consider the egalitarian welfare, namely the welfare of the worst-off agent. Since exact computation of the optimum is difficult \cite{}, we scale by $n$ instead of the exact optimum. 

Results are similar to the utilitarian case and show the non-competitiveness of RSD and Naive Boston, the positive effect of GTTC, and the overall superiority of TLSG and TLQG. In Figure~\ref{fig:egal_4} the lines for TLSG and TLQG are indistinguishable. We do not compare with PS because it is inherently random and so comparison would be unfair to the other algorithms --- the expectation of the minimum welfare is less than the same as the minimum of the expectations. Note that welfare increases with $n$ for the new algorithms, but decreases for the old ones.

\if01
\begin{figure}
\caption{Egalitarian welfare of our 8 basic algorithms}
\label{fig:egal_8}
\includegraphics[width=14cm]{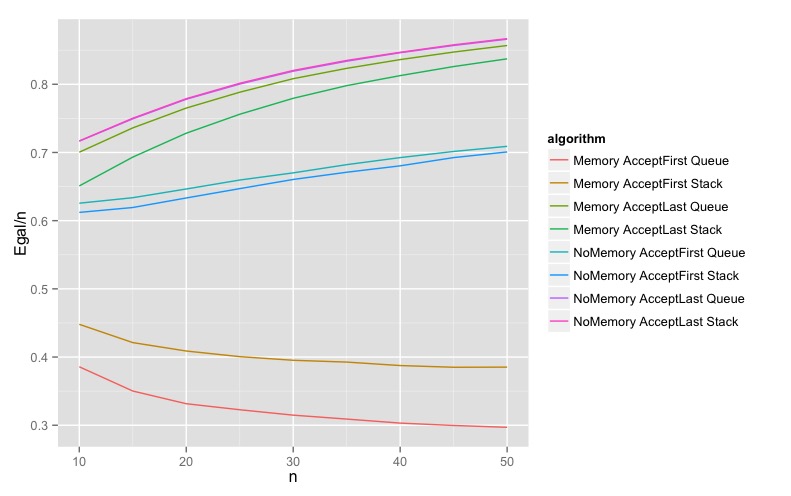}
\end{figure}

\begin{figure}
\caption{Egalitarian welfare improvement for Accept-Last with TTC}
\label{fig:egal_TTC}
\includegraphics[width=14cm]{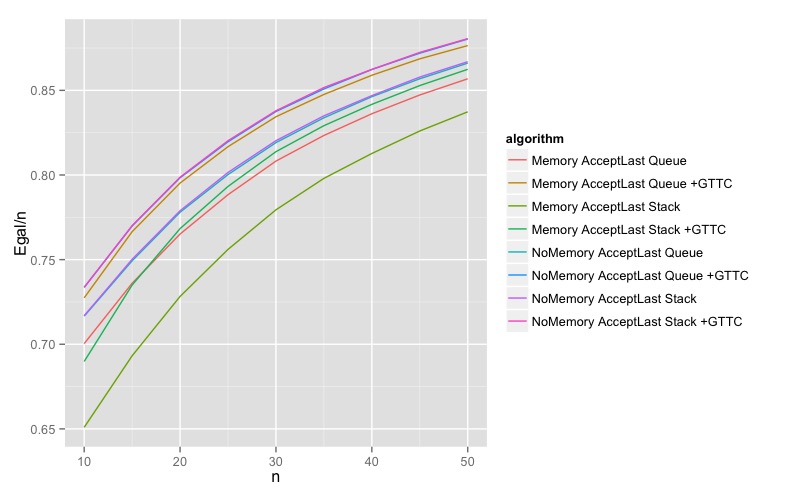}
\end{figure}
\fi

\begin{figure}
\caption{Normalized egalitarian welfare comparison}
\label{fig:egal_4}
\includegraphics[width=14cm]{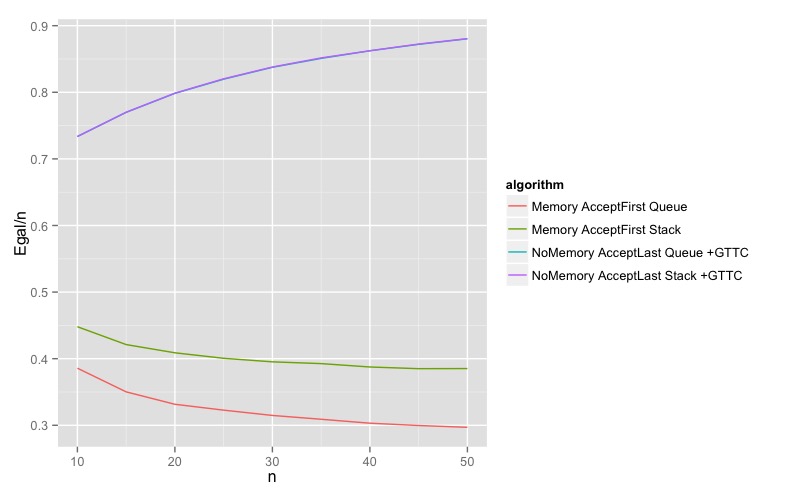}
\end{figure}

\subsubsection{Egalitarian welfare bounds}
\label{sss:prop k}

Say that an algorithm satisfies a \emph{conditional egalitarian welfare bound} of $k$ if every agent receives one of its top $k$ choices, \emph{whenever that is possible under some allocation}. If $k$ happens to be the minimal possible value, this says that the algorithm yields optimal egalitarian Borda welfare on that input. All reasonable algorithms (including all those in this paper) satisfy a bound of $1$, because if all agents have different top choices, each receives its top choice. We investigated the case $k=2$ and found that none of our algorithms satisfy it in general (details omitted). However, TLS satisfies the bound with $k=2$ when $n=3$, as does its queue-based counterpart TLQ, while none of the other algorithms does.  Thus for $n=3$ these algorithms are egalitarian-optimal.

\if01
\jl{each row is the probability of agent i getting their $j$th choice instead of getting item j}

I have got examples of each algorithm failing k=2.

a,b,c
a,b,c
a,c,b

PSR, AB, RSD, NMAFS, MALS, NB, NMAFQ, MALQ all give allocation:

1/3 1/2 1/6
1/3 1/2 1/6
1/3 2/3 0

when 1:a, 2:b, 3:c will have each agent getting top 2 choices.

NMALQ fails k=2 on :
a,b,c,d
a,b,c,d
a,c,d,b
b,d,a,c

allocation:
1/2 1/3 1/6 0
1/2 1/3 1/6 0
0 2/3 1/3/ 0
1/3 2/3 0 0

NMALS fails k=2 on
a,b,c,d
a,b,c,d
a,c,d,b
c,d,a,b

allocation:
1/2 1/2 0 0
1/2 1/2 0 0
0 3/4 1/4 0
1/4 3/4 0 0

The two NMAL algorithms passes k=2 for n=m=3.

\fi

\subsection{Order bias}
\label{ss:order}

In the case where all agents have the same preferences over items, some algorithms (such as serial dictatorship) are clearly biased toward the first agent while others (accept-last algorithms) are clearly biased toward the last agent. We define the \emph{order bias} of an algorithm to be the maximum over all pairs of agents of the difference of the expected (under the uniform distribution on preferences) Borda welfare gained, and normalize by $n$.

Results show that the order bias of queue-based algorithms is markedly smaller than that for stack-based algorithms. Our best welfare algorithms, namely TLSG and TLQG, have almost zero order bias (PS has zero order bias by definition), but the randomized Serial Dictatorship and Naive Boston algorithms have substantial order bias, with the former being clearly more biased than all other algorithms. Adding TTC to the Accept-Last algorithms substantially reduces order bias (not shown).

\begin{figure}
\caption{Normalized order bias of our 8 basic algorithms}
\label{fig:bias_8}
\includegraphics[width=14cm]{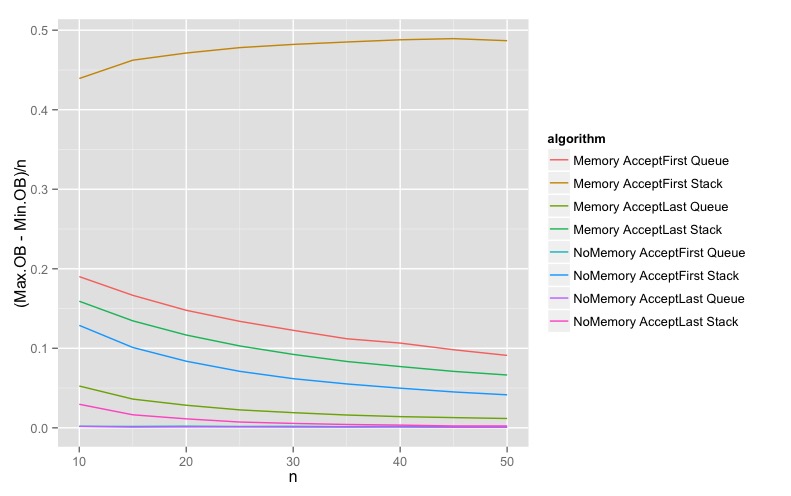}
\end{figure}

\if01
\begin{figure}
\caption{Normalized order bias improvement for Accept-Last with TTC}
\label{fig:bias_TTC}
\includegraphics[width=14cm]{bias_TTC.jpeg}
\end{figure}
\fi

\begin{figure}
\caption{Normalized order bias comparison}
\label{fig:bias_4}
\includegraphics[width=14cm]{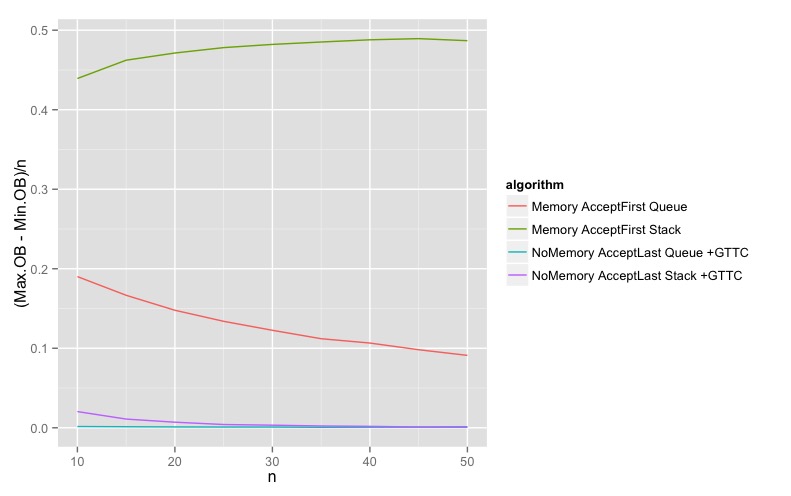}
\end{figure}

\section{Conclusion}
\label{s:conc}

\subsection{Summary of results}
\label{ss:summary}
We have introduced 10 new algorithms for 1-sided matching, none of them equivalent to each other or to any algorithms in the literature, to our knowledge. Their derivation using the Gale-Shapley framework gives a unified description. Each algorithm runs in worst-case time of order $n^2$ (permanent memory) or $n^3$ (temporary memory). Although they lack strong axiomatic foundations, several of the algorithms perform well on criteria such as egalitarian welfare, utilitarian welfare, and order bias, in sharp contrast to the standard Serial Dictatorship or Boston algorithms. The algorithms TLSG and TLQG perform better overall than the standard algorithms Serial Dictatorship, Naive Boston, and Probabilistic Serial on such measures.

In this article we have introduced a new (to our knowledge) performance criterion related to symmetry, namely order bias. The queue-based algorithms perform better overall than the stack-based ones by this measure. The algorithm TLQ (and TLQG) has remarkably small order bias. Our intuition is that the queue simulates randomization of the order of choosing by agents. 

The approach we use of considering algorithms that are not ex-post efficient, and then running TTC on their output, seems new to us. Ex-post efficiency is, roughly speaking, a local optimum criterion. By avoiding prematurely locking in efficiency, our new Accept-Last + TTC algorithms seem to be able to achieve efficient outcomes with higher global welfare.

The benefits of using temporary versus permanent memory are relatively small compared to the gains made by using TTC, for example. However they are real and for situations where solution quality is substantially more important than runtime, we recommend its use. The best overall algorithms in terms of solution quality are arguably TLSG and TLQG. The latter has lower order bias and the former higher welfare, although the differences between the algorithms are small.

Some of our new algorithms may be useful for specialized situations (in addition to party games). For example, TFS seems to treat all agents equally in welfare except the last, who has a definite advantage --- this may be useful, for example, when one agent is a small child. Having (almost) zero order bias is a strong fairness condition that may be very important in some applications. Our $1$-sided Naive Boston algorithm PFQ  maximizes the number of agents receiving their first choice. 

Enlarging the stock of basic algorithms has more benefits than simply allowing us 10 more algorithms to choose from. The concept of \emph{hybridization} has been used by Mennle \& Seuken \cite{MeSe2013}. This simply forms a new algorithm $(1-p)A + pB$ from random allocation algorithms A and B and a fixed $p\in [0,1]$ by taking the convex combination $(1-p)M_A + pM_B$ of the stochastic matrices output A and B. This allows us to trade off desirable properties such as strategyproofness and efficiency in a controlled way. Mennle \& Seuken considered only RSD, PS, Naive Boston and the algorithm maximizing utilitarian welfare as their basic algorithms. We believe that our new algorithms will prove useful as building blocks for hybrid algorithms with good overall behaviour.

\subsection{Future work}
\label{ss:future}

The Adaptive Boston school choice algorithm improves over Naive Boston by allowing agents to skip proposals that will obviously be rejected because a school has reached capacity. It satisfies a property intermediate between weak strategyproofness and strategyproofness, called \emph{partial strategyproofness} by Mennle \& Seuken \cite{MeSe2014}. Although we found a sequential interpretation of Naive Boston that avoided discussion of simultaneous proposals, we have not yet done this for Adaptive Boston. 

Each of our algorithms extends to the school choice situation. The ``memory" component and the ``data structure" component translate directly with no changes required. The aceptance policy if more complicated. In school choice, schools accept applicants provisionally until capacity is reached, and then each new temporary enrolment requires an existing enrolment to be cancelled. Our Accept-Last and Accept-First policies in the capacity 1 case described above could also be termed ``Reject-Current'' or ``Reject-New". In the general school choice situation we would need to create fictitious prefernces for schools to enable them to decide which current student to reject. Some obvious ways to do that include  FIFO or  LIFO. 

The case where the number of items exceeds the number of agents, and agents receive bundles of items chosen one at a time, is complicated. The order in which agents should choose in each round (here a round ends when all agents have incremented  their previous total of items by 1) must be specified. For example, when we have 2 agents with preferences $a>b>c>d$ over 4 items, and the picking sequence $1221$, Serial Dictatorship awards $a, d$ to agent 1 and $b,c$ to agent 2. However under PFQ agent 2 attempts to get item $a$ and fails, going to the back of the queue and hence missing that turn. In the next turn it is allocated $b$, then agent 1 tries for $b$ and fails. Thus the picking sequence must be extended. Accept-Last algorithms work better in this situation.

We leave further exploration of these interesting cases with $m\neq n$ to future work.
\bibliographystyle{alpha}
\bibliography{assignment}

\appendix
\appendixpage
\addappheadtotoc

\section{Sample executions of algorithms}
\label{apps:examples}
For a profile where agent 1, 2 and 3 have preferences $a > b > c > d$, and agent 4 have preferences $b > a > c > d$, and the initial ordering of agents being in ascending numerical order, the sequence of proposal and rejection using the various algorithms will be as follows.
\subsection{Permanent Memory, Early-Proposal preferred, Stack}
\label{appss:PMEPS}
\begin{tabular}[!ht]{|c|c|c|c|}
\hline
	Proposal & Outcome & Order of Remaining Agents & Current Partial Matching\\
\hline
    $a_1 =>$ a & tentatively matched & ${a_2, a_3, a_4}$ & $a_1$:a\\
    $a_2 =>$ a & $a_2$ rejected & ${a_2, a_3, a_4}$ & $a_1$:a\\
    $a_2 =>$ b & tentatively matched & ${a_3, a_4}$ & $a_1$:a, $a_2$:b\\
    $a_3 =>$ a & $a_3$ rejected & ${a_3, a_4}$ & $a_1$:a, $a_2$:b\\
    $a_3 =>$ b & $a_3$ rejected & ${a_3, a_4}$ & $a_1$:a, $a_2$:b\\
    $a_3 =>$ c & tentatively matched & ${a_4}$ & $a_1$:a, $a_2$:b, $a_3$:c\\
    $a_4 =>$ b & $a_4$ rejected & ${a_4}$ & $a_1$:a, $a_2$:b, $a_3$:c\\
    $a_4 =>$ a & $a_4$ rejected & ${a_4}$ & $a_1$:a, $a_2$:b, $a_3$:c\\
    $a_4 =>$ c & $a_4$ rejected & ${a_4}$ & $a_1$:a, $a_2$:b, $a_3$:c\\
    $a_4 =>$ d & tentatively matched & ${}$ & $a_1$:a, $a_2$:b, $a_3$:c, $a_4$:d\\
\hline
\end{tabular}
\subsection{Permanent Memory, Early-Proposal preferred, Queue}
\label{appss:PMEPQ}
\begin{tabular}[!ht]{|c|c|c|c|}
\hline
	Proposal & Outcome & Order of Remaining Agents & Current Partial Matching\\
\hline
    $a_1 =>$ a & tentatively matched & ${a_2, a_3, a_4}$ & $a_1$:a\\
    $a_2 =>$ a & $a_2$ rejected & ${a_3, a_4, a_2}$ & $a_1$:a\\
    $a_3 =>$ a & $a_3$ rejected & ${a_4, a_2, a_3}$ & $a_1$:a\\
    $a_4 =>$ b & tentatively matched & ${a_2, a_3}$ & $a_1$:a, $a_4$:b\\
    $a_2 =>$ b & $a_2$ rejected & ${a_3, a_2}$ & $a_1$:a, $a_4$:b\\
    $a_3 =>$ b & $a_3$ rejected & ${a_2, a_3}$ & $a_1$:a, $a_4$:b\\
    $a_2 =>$ c & tentatively matched & ${a_3}$ & $a_1$:a, $a_2$:c, $a_4$:b\\
    $a_3 =>$ c & $a_3$ rejected & ${a_3}$ & $a_1$:a, $a_2$:c, $a_4$:b\\
    $a_3 =>$ d & tentatively matched & ${}$ & $a_1$:a, $a_2$:c, $a_3$:d, $a_4$:b\\
\hline
\end{tabular}
\subsection{Permanent Memory, Last-Proposal preferred, Stack}
\label{appss:PMLPS}
\begin{tabular}[!ht]{|c|c|c|c|}
\hline
	Proposal & Outcome & Order of Remaining Agents & Current Partial Matching\\
\hline
    $a_1 =>$ a & tentatively matched & ${a_2, a_3, a_4}$ & $a_1$:a\\
    $a_2 =>$ a & $a_1$ rejected & ${a_1, a_3, a_4}$ & $a_2$:a\\
    $a_1 =>$ b & tentatively matched & ${a_3, a_4}$ & $a_1$:b, $a_2$:a\\
    $a_3 =>$ a & $a_2$ rejected & ${a_2, a_4}$ & $a_1$:b, $a_3$:a\\
    $a_2 =>$ b & $a_1$ rejected & ${a_1, a_4}$ & $a_2$:b, $a_3$:a\\
    $a_1 =>$ c & tentatively matched & ${a_4}$ & $a_1$:c, $a_2$:b, $a_3$:a\\
    $a_4 =>$ b & $a_2$ rejected & ${a_2}$ & $a_1$:c, $a_3$:a, $a_4$:b\\
    $a_2 =>$ c & $a_1$ rejected & ${a_1}$ & $a_2$:c, $a_3$:a, $a_4$:b\\
    $a_1 =>$ d & tentatively matched & ${}$ & $a_1$:d, $a_2$:c, $a_3$:a, $a_4$:b\\
\hline
\end{tabular}
\subsection{Permanent Memory, Last-Proposal preferred, Queue}
\label{appss:PMLPQ}
\begin{tabular}[!ht]{|c|c|c|c|}
\hline
	Proposal & Outcome & Order of Remaining Agents & Current Partial Matching\\
\hline
    $a_1 =>$ a & tentatively matched & ${a_2, a_3, a_4}$ & $a_1$:a\\
    $a_2 =>$ a & $a_1$ rejected & ${a_3, a_4, a_1}$ & $a_2$:a\\
    $a_3 =>$ a & $a_2$ rejected & ${a_4, a_1, a_2}$ & $a_3$:a\\
    $a_4 =>$ b & tentatively matched & ${a_1, a_2}$ & $a_3$:a, $a_4$:b\\
    $a_1 =>$ b & $a_4$ rejected & ${a_2, a_4}$ & $a_1$:b, $a_3$:a\\
    $a_2 =>$ b & $a_1$ rejected & ${a_4, a_1}$ & $a_2$:b, $a_3$:a\\
    $a_4 =>$ a & $a_3$ rejected & ${a_1, a_3}$ & $a_2$:b, $a_4$:a\\
	$a_1 =>$ c & tentatively matched & ${a_3}$ & $a_1$:c, $a_2$:b, $a_4$:a\\
	$a_3 =>$ c & $a_1$ rejected & ${a_1}$ & $a_2$:b, $a_3$:c, $a_4$:a\\
    $a_1 =>$ d & tentatively matched & ${}$ & $a_1$:d, $a_2$:c, $a_3$:b, $a_4$:a\\
\hline
\end{tabular}
\subsection{Temporary Memory, Early-Proposal preferred, Stack}
\label{appss:TMEPS}

\begin{tabular}[!ht]{|c|c|c|c|c|}
\hline
	Proposal & Outcome & Order of Remaining Agents & Current Partial Matching & Item Preferences\\
\hline
    $a_1 =>$ a & tentatively matched & ${a_2, a_3, a_4}$ & $a_1$:a & none\\
    $a_2 =>$ a & $a_1$ rejected & ${a_1, a_3, a_4}$ & $a_2$:a & $i_a$:$a_2 > a_1$\\
    $a_1 =>$ b & tentatively matched & ${a_3, a_4}$ & $a_1$:b, $a_2$:a & none\\
    $a_3 =>$ a & $a_2$ rejected & ${a_2, a_4}$ & $a_1$:b, $a_3$:a & $i_a$:$a_3 > a_2$\\
    $a_2 =>$ a & $a_2$ rejected & ${a_2, a_4}$ & $a_1$:b, $a_3$:a & $i_a$:$a_3 > a_2$\\
    $a_2 =>$ b & $a_1$ rejected & ${a_1, a_4}$ & $a_2$:b, $a_3$:a & $i_a$:$a_3 > a_2$\\
    & & & & $i_b$:$a_2 > a_1$\\
    $a_1 =>$ a & $a_1$ rejected & ${a_1, a_4}$ & $a_2$:b, $a_3$:a & $i_a$:$a_3 > a_2 > a_1$\\
    & & & & $i_b$:$a_2 > a_1$\\
    $a_1 =>$ b & $a_1$ rejected & ${a_1, a_4}$ & $a_2$:b, $a_3$:a & $i_a$:$a_3 > a_2 > a_1$\\
    & & & & $i_b$:$a_2 > a_1$\\
    $a_1 =>$ c & tentatively matched & ${a_4}$ & $a_1$:c, $a_2$:b, $a_3$:a & none\\
    $a_4 =>$ b & $a_2$ rejected & ${a_2}$ & $a_1$:c, $a_3$:a, $a_4$:b & $i_b$:$a_4 > a_2$\\
    $a_2 =>$ a & $a_3$ rejected & ${a_3}$ & $a_1$:c, $a_2$:a, $a_4$:b & $i_a$:$a_2 > a_3$\\
    & & & & $i_b$:$a_4 > a_2$\\
    $a_3 =>$ a & $a_3$ rejected & ${a_3}$ & $a_1$:c, $a_2$:a, $a_4$:b & $i_a$:$a_2 > a_3$\\
    & & & & $i_b$:$a_4 > a_2$\\
    $a_3 =>$ b & $a_3$ rejected & ${a_4}$ & $a_1$:c, $a_2$:a, $a_4$:b & $i_a$:$a_2 > a_3$\\
    & & & & $i_b$:$a_4 > a_2 > a_3$\\
    $a_3 =>$ c & $a_1$ rejected & ${a_1}$ & $a_2$:a, $a_3$:c, $a_4$:b & $i_a$:$a_2 > a_3$\\
    & & & & $i_b$:$a_4 > a_2$\\
    & & & & $i_c$:$a_3 > a_1$\\
    $a_1 =>$ a & $a_1$ rejected & ${a_1}$ & $a_2$:a, $a_3$:c, $a_4$:b & $i_a$:$a_2 > a_3 > a_1$\\
    & & & & $i_b$:$a_4 > a_2$\\
    & & & & $i_c$:$a_3 > a_1$\\
    $a_1 =>$ b & $a_1$ rejected & ${a_1}$ & $a_2$:a, $a_3$:c, $a_4$:b & $i_a$:$a_2 > a_3 > a_1$\\
    & & & & $i_b$:$a_4 > a_2 > a_1$\\
    & & & & $i_c$:$a_3 > a_1$\\
    $a_1 =>$ c & $a_1$ rejected & ${a_1}$ & $a_2$:a, $a_3$:c, $a_4$:b & $i_a$:$a_2 > a_3 > a_1$\\
    & & & & $i_b$:$a_4 > a_2 > a_1$\\
    & & & & $i_c$:$a_3 > a_1$\\
    $a_1 =>$ d & tentatively matched & ${}$ & $a_1$:d, $a_2$:a, $a_3$:c, $a_4$:b & none\\
\hline
\end{tabular}
\subsection{Temporary Memory, Early-Proposal preferred, Queue}
\label{appss:TMEPQ}

\begin{longtable}[!ht]{|c|c|c|c|c|}
\hline
	Proposal & Outcome & Order of Remaining Agents & Current Partial Matching & Item Preferences\\
\hline
    $a_1 =>$ a & tentatively matched & ${a_2, a_3, a_4}$ & $a_1$:a & none\\
    $a_2 =>$ a & $a_1$ rejected & ${a_3, a_4, a_1}$ & $a_2$:a & $i_a$:$a_2 > a_1$\\
    $a_3 =>$ a & $a_3$ rejected & ${a_4, a_1, a_3}$ & $a_2$:a & $i_a$:$a_2 > a_1 > a_3$\\
    $a_4 =>$ b & tentatively matched & ${a_1, a_3}$ & $a_2$:a, $a_4$:b & none\\
    $a_1 =>$ a & $a_2$ rejected & ${a_3, a_2}$ & $a_1$:a, $a_4$:b & $i_a$:$a_1 > a_2$\\
    $a_3 =>$ a & $a_3$ rejected & ${a_2, a_3}$ & $a_1$:a, $a_4$:b & $i_a$:$a_1 > a_2 > a_3$\\
    $a_2 =>$ a & $a_2$ rejected & ${a_3, a_2}$ & $a_1$:a, $a_4$:b & $i_a$:$a_1 > a_2 > a_3$\\
    $a_3 =>$ b & $a_4$ rejected & ${a_2, a_4}$ & $a_1$:a, $a_3$:b & $i_a$:$a_2 > a_1 > a_3$\\
    & & & & $i_b$:$a_3 > a_4$\\
    $a_2 =>$ b & $a_2$ rejected & ${a_4, a_2}$ & $a_1$:a, $a_3$:b & $i_a$:$a_2 > a_1 > a_3$\\
    & & & & $i_b$:$a_3 > a_4 > a_2$\\
    $a_4 =>$ a & $a_4$ rejected & ${a_2, a_4}$ & $a_1$:a, $a_3$:b & $i_a$:$a_2 > a_1 > a_3 > a_4$\\
    & & & & $i_b$:$a_3 > a_4 > a_2$\\
    $a_2 =>$ c & tentatively matched & ${a_4}$ & $a_1$:a, $a_2$:c, $a_3$:b & none\\
    $a_4 =>$ b & $a_3$ rejected & ${a_3}$ & $a_1$:a, $a_2$:c, $a_4$:b & $i_b$:$a_4 > a_3$\\
    $a_3 =>$ a & $a_1$ rejected & ${a_1}$ & $a_2$:c, $a_3$:a, $a_4$:b & $i_a$:$a_3 > a_1$\\
    & & & & $i_b$:$a_4 > a_3$\\
    $a_1 =>$ a & $a_1$ rejected & ${a_1}$ & $a_2$:c, $a_3$:a, $a_4$:b & $i_a$:$a_3 > a_1$\\
    & & & & $i_b$:$a_4 > a_3$\\
    $a_1 =>$ b & $a_1$ rejected & ${a_1}$ & $a_2$:c, $a_3$:a, $a_4$:b & $i_a$:$a_3 > a_1$\\
    & & & & $i_b$:$a_4 > a_3 > a_1$\\
    $a_1 =>$ c & $a_2$ rejected & ${a_2}$ & $a_1$:c, $a_3$:a, $a_4$:b & $i_a$:$a_3 > a_1$\\
    & & & & $i_b$:$a_4 > a_3 > a_1$\\
    & & & & $i_c$:$a_1 > a_2$\\
    $a_2 =>$ a & $a_2$ rejected & ${a_2}$ & $a_1$:c, $a_3$:a, $a_4$:b & $i_a$:$a_3 > a_1 > a_2$\\
    & & & & $i_b$:$a_4 > a_3 > a_1$\\
    & & & & $i_c$:$a_1 > a_2$\\
    $a_2 =>$ b & $a_2$ rejected & ${a_2}$ & $a_1$:c, $a_3$:a, $a_4$:b & $i_a$:$a_3 > a_1 > a_2$\\
    & & & & $i_b$:$a_4 > a_3 > a_1 > a_2$\\
    & & & & $i_c$:$a_1 > a_2$\\
    $a_2 =>$ c & $a_2$ rejected & ${a_2}$ & $a_1$:c, $a_3$:a, $a_4$:b & $i_a$:$a_3 > a_1 > a_2$\\
    & & & & $i_b$:$a_4 > a_3 > a_1 > a_2$\\
    & & & & $i_c$:$a_1 > a_2$\\
    $a_2 =>$ d & tentatively matched & $a_1$:c, $a_2$:d, $a_3$:a, $a_4$:b & none\\
    $a_3 =>$ b & $a_1$ rejected & ${a_4, a_1}$ & $a_2$:a, $a_3$:b & $i_a$:$a_2 > a_1 > a_3$\\
    & & & & $i_b$:$a_3 > a_1 > a_4$\\
    $a_4 =>$ b & $a_4$ rejected & ${a_1, a_4}$ & $a_2$:a, $a_3$:b & $i_a$:$a_2 > a_1 > a_3$\\
    & & & & $i_b$:$a_3 > a_1 > a_4$\\
    $a_1 =>$ c & tentatively matched & ${a_4}$ & $a_1$:c, $a_3$:b, $a_4$:a & none\\
    $a_4 =>$ b & $a_3$ rejected & ${a_3}$ & $a_1$:c, $a_2$:a, $a_4$:b & $i_b$:$a_4 > a_3$\\
    $a_3 =>$ a & $a_2$ rejected & ${a_2}$ & $a_1$:c, $a_3$:a, $a_4$:b & $i_a$:$a_3 > a_2$\\
    & & & & $i_b$:$a_4 > a_3$\\
    $a_2 =>$ a & $a_2$ rejected & ${a_2}$ & $a_1$:c, $a_3$:a, $a_4$:b & $i_a$:$a_3 > a_2$\\
    & & & & $i_b$:$a_4 > a_3$\\
    $a_2 =>$ b & $a_4$ rejected & ${a_4}$ & $a_1$:c, $a_2$:b, $a_3$:a & $i_a$:$a_3 > a_2$\\
    & & & & $i_b$:$a_2 > a_4 > a_3$\\
    $a_4 =>$ a & $a_3$ rejected & ${a_4}$ & $a_1$:c, $a_2$:b, $a_4$:a & $i_a$:$a_4 > a_3 > a_2$\\
    & & & & $i_b$:$a_2 > a_4 > a_3$\\
    $a_3 =>$ b & $a_3$ rejected & ${a_3}$ & $a_1$:c, $a_2$:b, $a_4$:a & $i_a$:$a_4 > a_3 > a_2$\\
    & & & & $i_b$:$a_2 > a_4 > a_3$\\
    $a_3 =>$ c & $a_1$ rejected & ${a_1}$ & $a_2$:b, $a_3$:c, $a_4$:a & $i_a$:$a_4 > a_3 > a_2$\\
    & & & & $i_b$:$a_2 > a_4 > a_3$\\
    & & & & $i_c$:$a_3 > a_1$\\
    $a_1 =>$ a & $a_4$ rejected & ${a_4}$ & $a_1$:a, $a_2$:b, $a_3$:c & $i_a$:$a_1 > a_4 > a_3 > a_2$\\
    & & & & $i_b$:$a_2 > a_4 > a_3$\\
    & & & & $i_c$:$a_3 > a_1$\\
    $a_4 =>$ c & $a_3$ rejected & ${a_3}$ & $a_1$:a, $a_2$:b, $a_4$:c & $i_a$:$a_1 > a_4 > a_3 > a_2$\\
    & & & & $i_b$:$a_2 > a_4 > a_3$\\
    & & & & $i_c$:$a_4 > a_3 > a_1$\\
    $a_3 =>$ d & tentatively matched & ${}$ & $a_1$:a, $a_2$:b, $a_3$:d, $a_4$:c & none\\
\hline
\end{longtable}
\subsection{Temporary Memory, Last-Proposal preferred, Stack}
\label{appss:TMLPS}

\begin{tabular}[!ht]{|c|c|c|c|c|}
\hline
	Proposal & Outcome & Order of Remaining Agents & Current Partial Matching & Item Preferences\\
\hline
    $a_1 =>$ a & tentatively matched & ${a_2, a_3, a_4}$ & $a_1$:a & none\\
    $a_2 =>$ a & $a_1$ rejected & ${a_1, a_3, a_4}$ & $a_2$:a & $i_a$:$a_2 > a_1$\\
    $a_1 =>$ b & tentatively matched & ${a_3, a_4}$ & $a_1$:b, $a_2$:a & none\\
    $a_3 =>$ a & $a_2$ rejected & ${a_2, a_4}$ & $a_1$:b, $a_3$:a & $i_a$:$a_3 > a_2$\\
    $a_2 =>$ b & $a_1$ rejected & ${a_1, a_4}$ & $a_2$:b, $a_3$:a & $i_a$:$a_3 > a_2$\\
    & & & & $i_b$:$a_2 > a_1$\\
    $a_1 =>$ a & $a_3$ rejected & ${a_3, a_4}$ & $a_1$:a, $a_2$:b & $i_a$:$a_1 > a_3 > a_2$\\
    & & & & $i_b$:$a_2 > a_1$\\
    $a_3 =>$ b & $a_2$ rejected & ${a_2, a_4}$ & $a_1$:a, $a_3$:b & $i_a$:$a_1 > a_3 > a_2$\\
    & & & & $i_b$:$a_3 > a_2 > a_1$\\
    $a_2 =>$ c & tentatively matched & ${a_4}$ & $a_1$:a, $a_2$:c, $a_3$:b & none\\
    $a_4 =>$ b & $a_3$ rejected & ${a_3}$ & $a_1$:a, $a_2$:c, $a_4$:b & $i_b$:$a_4 > a_3$\\
    $a_3 =>$ a & $a_1$ rejected & ${a_1}$ & $a_2$:c, $a_3$:a, $a_4$:b & $i_a$:$a_3 > a_1$\\
    & & & & $i_b$:$a_4 > a_3$\\
    $a_1 =>$ a & $a_1$ rejected & ${a_1}$ & $a_2$:c, $a_3$:a, $a_4$:b & $i_a$:$a_3 > a_1$\\
    & & & & $i_b$:$a_4 > a_3$\\
    $a_1 =>$ b & $a_4$ rejected & ${a_4}$ & $a_1$:b, $a_2$:c, $a_3$:a & $i_a$:$a_3 > a_1$\\
    & & & & $i_b$:$a_1 > a_4 > a_3$\\
    $a_4 =>$ a & $a_3$ rejected & ${a_3}$ & $a_1$:b, $a_2$:c, $a_4$:a & $i_a$:$a_4 > a_3 > a_1$\\
    & & & & $i_b$:$a_1 > a_4 > a_3$\\
    $a_3 =>$ b & $a_3$ rejected & ${a_3}$ & $a_1$:b, $a_2$:c, $a_4$:a & $i_a$:$a_4 > a_3 > a_1$\\
    & & & & $i_b$:$a_1 > a_4 > a_3$\\
    $a_3 =>$ c & $a_2$ rejected & ${a_2}$ & $a_1$:b, $a_3$:c, $a_4$:a & $i_a$:$a_4 > a_3 > a_1$\\
    & & & & $i_b$:$a_1 > a_4 > a_3$\\
    & & & & $i_c$:$a_3 > a_2$\\
    $a_2 =>$ a & $a_4$ rejected & ${a_4}$ & $a_1$:b, $a_2$:a, $a_3$:c & $i_a$:$a_2 > a_4 > a_3 > a_1$\\
    & & & & $i_b$:$a_1 > a_4 > a_3$\\
    & & & & $i_c$:$a_3 > a_2$\\
    $a_4 =>$ c & $a_3$ rejected & ${a_3}$ & $a_1$:b, $a_2$:a, $a_4$:c & $i_a$:$a_2 > a_4 > a_3 > a_1$\\
    & & & & $i_b$:$a_1 > a_4 > a_3$\\
    & & & & $i_c$:$a_3 > a_2$\\
    $a_3 =>$ d & tentatively matched & ${}$ & $a_1$:b, $a_2$:a, $a_3$:d, $a_4$:c & none\\
\hline
\end{tabular}
\subsection{Temporary Memory, Last-Proposal preferred, Queue}
\label{appss:TMLPQ}

\begin{tabular}[!ht]{|c|c|c|c|c|}
\hline
	Proposal & Outcome & Order of Remaining Agents & Current Partial Matching & Item Preferences\\
\hline
    $a_1 =>$ a & tentatively matched & ${a_2, a_3, a_4}$ & $a_1$:a & none\\
    $a_2 =>$ a & $a_1$ rejected & ${a_3, a_4, a_1}$ & $a_2$:a & $i_a$:$a_2 > a_1$\\
    $a_3 =>$ a & $a_2$ rejected & ${a_4, a_1, a_2}$ & $a_3$:a & $i_a$:$a_3 > a_2 > a_1$\\
    $a_4 =>$ b & tentatively matched & ${a_1, a_2}$ & $a_3$:a, $a_4$:b & none\\
    $a_1 =>$ a & $a_3$ rejected & ${a_2, a_3}$ & $a_1$:a, $a_4$:b & $i_a$:$a_1 > a_3$\\
    $a_2 =>$ a & $a_1$ rejected & ${a_3, a_1}$ & $a_2$:a, $a_4$:b & $i_a$:$a_2 > a_1 > a_3$\\
    $a_3 =>$ a & $a_3$ rejected & ${a_1, a_3}$ & $a_2$:a, $a_4$:b & $i_a$:$a_2 > a_1 > a_3$\\
    $a_1 =>$ b & $a_4$ rejected & ${a_3, a_4}$ & $a_1$:b, $a_2$:a & $i_a$:$a_2 > a_1 > a_3$\\
    & & & & $i_b$:$a_1 > a_4$\\
    $a_3 =>$ b & $a_1$ rejected & ${a_4, a_1}$ & $a_2$:a, $a_3$:b & $i_a$:$a_2 > a_1 > a_3$\\
    & & & & $i_b$:$a_3 > a_1 > a_4$\\
    $a_4 =>$ b & $a_4$ rejected & ${a_1, a_4}$ & $a_2$:a, $a_3$:b & $i_a$:$a_2 > a_1 > a_3$\\
    & & & & $i_b$:$a_3 > a_1 > a_4$\\
    $a_1 =>$ c & tentatively matched & ${a_4}$ & $a_1$:c, $a_3$:b, $a_4$:a & none\\
    $a_4 =>$ b & $a_3$ rejected & ${a_3}$ & $a_1$:c, $a_2$:a, $a_4$:b & $i_b$:$a_4 > a_3$\\
    $a_3 =>$ a & $a_2$ rejected & ${a_2}$ & $a_1$:c, $a_3$:a, $a_4$:b & $i_a$:$a_3 > a_2$\\
    & & & & $i_b$:$a_4 > a_3$\\
    $a_2 =>$ a & $a_2$ rejected & ${a_2}$ & $a_1$:c, $a_3$:a, $a_4$:b & $i_a$:$a_3 > a_2$\\
    & & & & $i_b$:$a_4 > a_3$\\
    $a_2 =>$ b & $a_4$ rejected & ${a_4}$ & $a_1$:c, $a_2$:b, $a_3$:a & $i_a$:$a_3 > a_2$\\
    & & & & $i_b$:$a_2 > a_4 > a_3$\\
    $a_4 =>$ a & $a_3$ rejected & ${a_4}$ & $a_1$:c, $a_2$:b, $a_4$:a & $i_a$:$a_4 > a_3 > a_2$\\
    & & & & $i_b$:$a_2 > a_4 > a_3$\\
    $a_3 =>$ b & $a_3$ rejected & ${a_3}$ & $a_1$:c, $a_2$:b, $a_4$:a & $i_a$:$a_4 > a_3 > a_2$\\
    & & & & $i_b$:$a_2 > a_4 > a_3$\\
    $a_3 =>$ c & $a_1$ rejected & ${a_1}$ & $a_2$:b, $a_3$:c, $a_4$:a & $i_a$:$a_4 > a_3 > a_2$\\
    & & & & $i_b$:$a_2 > a_4 > a_3$\\
    & & & & $i_c$:$a_3 > a_1$\\
    $a_1 =>$ a & $a_4$ rejected & ${a_4}$ & $a_1$:a, $a_2$:b, $a_3$:c & $i_a$:$a_1 > a_4 > a_3 > a_2$\\
    & & & & $i_b$:$a_2 > a_4 > a_3$\\
    & & & & $i_c$:$a_3 > a_1$\\
    $a_4 =>$ c & $a_3$ rejected & ${a_3}$ & $a_1$:a, $a_2$:b, $a_4$:c & $i_a$:$a_1 > a_4 > a_3 > a_2$\\
    & & & & $i_b$:$a_2 > a_4 > a_3$\\
    & & & & $i_c$:$a_4 > a_3 > a_1$\\
    $a_3 =>$ d & tentatively matched & ${}$ & $a_1$:a, $a_2$:b, $a_3$:d, $a_4$:c & none\\
\hline
\end{tabular}

\subsection{}
\label{app:alldiff}

\begin{prop} The randomized versions of each of the 8 algorithms introduced in this article are all inequivalent.
\end{prop}
\begin{proof}
For a profile where agents 1, 2 and 3 have preferences $a > b > c > d$, and agent 4 has preferences $a > c > d > b$, the algorithms make the following allocations.

\begin{tabular}[!ht]{|c|c|c|c|c|}
\hline
	Alias & Memory & Item Preference & Data Structure & Probabilistic Allocation\\
\hline
	Random Serial & Permanent  & Accept-First & Stack & $1/4$, $1/3$, $1/6$, $1/4$\\
    Dictatorship&&&& $1/4$, $1/3$, $1/6$, $1/4$\\
    &&&& $1/4$, $1/3$, $1/6$, $1/4$\\
    &&&& $1/4$, \space$0$\space, $1/2$, $1/4$\\
    Naive Boston & Permanent& Accept-First & Queue & $1/4$, $1/3$, $1/12$, $1/3$\\
    &&&& $1/4$, $1/3$, $1/12$, $1/3$\\
	&&&& $1/4$, $1/3$, $1/12$, $1/3$\\
    &&&& $1/4$, \space$0$\space, $3/4$, \space$0$\space\\
	&Memory & Accept-Last & Stack & $1/4$, $1/3$, $1/4$, $1/6$\\
    &&&& $1/4$, $1/3$, $1/4$, $1/6$\\
    &&&& $1/4$, $1/3$, $1/4$, $1/6$\\
    &&&& $1/4$, \space$0$\space, $1/4$, $1/2$\\
    &Memory & Accept-Last & Queue & $1/4$, $1/3$, $1/3$, $1/12$\\
    &&&& $1/4$, $1/3$, $1/3$, $1/12$\\
    &&&& $1/4$, $1/3$, $1/3$, $1/12$\\
    &&&& $1/4$, \space$0$\space, \space$0$\space, $3/4$\\
	Iterative Dictatorship&Temporary &Accept-First & Stack & $1/4$, $1/3$, $1/4$, $1/6$\\
    &&&& $1/4$, $1/3$, $1/4$, $1/6$\\
    &&&& $1/4$, $1/3$, $1/4$, $1/6$\\
    &&&& $1/4$, \space$0$\space, $1/4$, $1/2$\\
    &Temporary & Accept-First & Queue & $1/3$, $1/3$, $1/4$, $1/12$\\
    &&&& $1/3$, $1/3$, $1/4$, $1/12$\\
    &&&& $1/3$, $1/3$, $1/4$, $1/12$\\
    &&&& \space$0$\space, \space$0$\space, $1/4$, $3/4$\\
	Yankee Swap & Temporary  & Accept-Last & Stack & $1/12$, $1/3$, $1/3$, $1/4$\\
    &&&& $1/12$, $1/3$, $1/3$, $1/4$\\
    &&&& $1/12$, $1/3$, $1/3$, $1/4$\\
    &&&& $3/4$, \space$0$\space, \space$0$\space, $1/4$\\
    &Temporary  & Accept-Last & Queue & $1/12$, $1/3$, $1/3$, $1/4$\\
    &&&& $1/12$, $1/3$, $1/3$, $1/4$\\
    &&&& $1/12$, $1/3$, $1/3$, $1/4$\\
    &&&& $3/4$, \space$0$\space, \space$0$\space, $1/4$\\
\hline
\end{tabular}
\end{proof}

\end{document}